\documentclass[
paper=letter,%
numbers=noendperiod,%
captions=nooneline,%
abstracton,%
DIV=10%
]{scrartcl}\usepackage[]{graphicx}\usepackage[]{color}
\makeatletter
\def\maxwidth{ %
  \ifdim\Gin@nat@width>\linewidth
    \linewidth
  \else
    \Gin@nat@width
  \fi
}
\makeatother

\usepackage{Sweave}

\usepackage[T1]{fontenc}%
\usepackage{lmodern}%
\usepackage[american]{babel}%
\usepackage{microtype}%

\usepackage[
hyperref,%
table%
]{xcolor}%

\usepackage{scrlayer-scrpage}%

\usepackage{eso-pic}%
\usepackage{rotating}%
\usepackage{amsmath}%
\usepackage{mathtools}%
\usepackage{amssymb}%
\usepackage{amsthm}%
\usepackage{thmtools}%
\usepackage{etoolbox}%
\usepackage{bm}%
\usepackage{bbm}%
\usepackage{enumitem}%
\usepackage{graphicx}%
\usepackage{grffile}%
\usepackage{tikz}%
\usepackage{wrapfig}%
\usepackage{tabularx}%
\usepackage{siunitx}%
\usepackage{booktabs}%
\usepackage{multirow}%
\usepackage{vruler}%
\usepackage{fancyvrb}%
\usepackage{textcomp}%
\usepackage{listings}%
\usepackage{csquotes}%
\usepackage[
style=authoryear,%
dashed=false,%
hyperref=true,%
useprefix=true,%
maxnames=2,%
maxbibnames=6,%
uniquename=false%
]{biblatex}%
\usepackage[
hypertexnames=false,%
setpagesize=false,%
pdfborder={0 0 0},%
pdfstartview=Fit,%
bookmarksopen=true,%
bookmarksnumbered=true%
]{hyperref}%
\setkomafont{pageheadfoot}{\normalfont\normalcolor\sffamily}%
\setkomafont{pagenumber}{\normalfont\normalcolor\sffamily}%
\automark{section}%
\setkomafont{captionlabel}{\normalfont\normalcolor\sffamily\bfseries}%

\setlist{%
  align=left,%
  labelsep=*,%
  leftmargin=*,%
  listparindent=\parindent,%
  parsep=0pt,%
  topsep=1mm,%
  itemsep=1mm%
}
\newcommand*{\mysquare}{\rule[0.18em]{0.36em}{0.36em}}
\newcommand*{\mytriangle}{\raisebox{0.12em}{\resizebox{0.48em}{0.48em}{$\blacktriangleright$}}}
\newcommand*{\mybar}{\rule[0.32em]{0.62em}{0.08em}}
\newcommand*{\mydot}{\raisebox{0.14em}{\resizebox{0.44em}{!}{$\bullet$}}}
\setlist[itemize,1]{label={\mysquare\ }}%
\setlist[itemize,2]{label={\mytriangle\ }}%
\setlist[itemize,3]{label={\mybar\ }}%
\setlist[itemize,4]{label={\mydot\ }}%
\setlist[enumerate,1]{label=\arabic*)}%
\setlist[enumerate,2]{label=\arabic{enumi}.\arabic*)}%
\setlist[enumerate,3]{label=\arabic{enumi}.\arabic{enumii}.\arabic*)}%

\makeatletter
\newcommand\myisodate{\number\year-\ifcase\month\or 01\or 02\or 03\or 04\or 05\or 06\or 07\or 08\or 09\or 10\or 11\or 12\fi-\ifcase\day\or 01\or 02\or 03\or 04\or 05\or 06\or 07\or 08\or 09\or 10\or 11\or 12\or 13\or 14\or 15\or 16\or 17\or 18\or 19\or 20\or 21\or 22\or 23\or 24\or 25\or 26\or 27\or 28\or 29\or 30\or 31\fi}%
\makeatother
\newcommand*{\abstractnoindent}{}%
\let\abstractnoindent\abstract
\renewcommand*{\abstract}{\let\quotation\quote\let\endquotation\endquote
  \abstractnoindent}
\deffootnote[1em]{1em}{1em}{\textsuperscript{\thefootnotemark}}%
\xdefinecolor{darkgray}{RGB}{56, 56, 56}%
\xdefinecolor{semilightgray}{RGB}{240, 240, 240}%
\xdefinecolor{lightgray}{RGB}{247, 247, 247}%
\xdefinecolor{blue}{RGB}{58, 95, 205}%
\xdefinecolor{red}{RGB}{205, 41, 144}%
\xdefinecolor{orange}{RGB}{238, 118, 0}%
\xdefinecolor{chocolate}{RGB}{205, 102, 29}%
\lstdefinestyle{input}{
  backgroundcolor=\color{semilightgray},%
  commentstyle=\itshape\color{black},%
  keywordstyle=\color{black},%
  emphstyle=\color{black},%
  stringstyle=\color{black},%
  numbers=left,%
  numbersep=4.8pt,%
  numberstyle=\color{darkgray!80}\tiny%
}
\lstdefinestyle{output}{
  backgroundcolor=\color{lightgray}%
}

\lstdefinestyle{Rstyle}{
  language=R,%
  keywords={function, if, else, switch, repeat, while, for, in, next, break},%
  otherkeywords={},%
  emph={TRUE, FALSE, NULL, NA, NaN, Inf}%
}
\expandafter\let\csname Sinput\endcsname\relax
\expandafter\let\csname endSinput\endcsname\relax
\expandafter\let\csname Soutput\endcsname\relax
\expandafter\let\csname endSoutput\endcsname\relax
\lstnewenvironment{Sinput}[1][]{%
  \lstset{style=input, style=Rstyle}
  #1%
}{\vspace{-0.25\baselineskip}}%
\lstnewenvironment{Soutput}[1][]{%
  \lstset{style=output, style=Rstyle}
  #1%
}{\vspace{-0.25\baselineskip}}%

\lstdefinestyle{LaTeXstyle}{
  language=[LaTeX]TeX,%
  texcs={},%
  otherkeywords={}%
}
\lstnewenvironment{LaTeXinput}[1][]{%
  \lstset{style=input, style=LaTeXstyle}
  #1%
}{\vspace{-0.25\baselineskip}}%
\lstnewenvironment{LaTeXoutput}[1][]{%
  \lstset{style=output, style=LaTeXstyle}
  #1%
}{\vspace{-0.25\baselineskip}}%

\lstdefinestyle{otherstyle}{
  language={},%
  otherkeywords={},%
  upquote=true%
}
\lstnewenvironment{otherinput}[1][]{%
  \lstset{style=input, style=otherstyle}
  #1%
}{\vspace{-0.25\baselineskip}}%
\lstnewenvironment{otheroutput}[1][]{%
  \lstset{style=output, style=otherstyle}
  #1%
}{\vspace{-0.25\baselineskip}}%
\newcommand*{\code}{\lstinline[
  basicstyle=\ttfamily,
  style=otherstyle,
  literate={~}{{$\sim$}}1
  ]}

\DeclareNameAlias{sortname}{last-first}%
\DefineBibliographyExtras{american}{\DeclareQuotePunctuation{}}%
\renewbibmacro*{volume+number+eid}{%
  \setunit*{\addcomma\space}%
  \printfield{volume}%
  \printfield{number}}
\DeclareFieldFormat*{number}{(#1)}
\DeclareFieldFormat*{title}{#1}%
\DeclareFieldFormat{doi}{%
  \ifhyperref
    {\href{http://dx.doi.org/#1}{\nolinkurl{doi:#1}}}%
    {\nolinkurl{doi:#1}}}%
\renewbibmacro*{in:}{}%
\DeclareFieldFormat{isbn}{ISBN #1}%
\DeclareFieldFormat{pages}{#1}%
\DeclareFieldFormat{url}{\url{#1}}%
\DeclareFieldFormat{urldate}{\mkbibparens{#1}}%
\renewcommand*{\cite}[2][]{\textcite[#1]{#2}}%

\newif\ifstarttheorem
\declaretheoremstyle[%
  spaceabove=0.5em,
  spacebelow=0.5em,
  headfont=\sffamily\bfseries\global\starttheoremtrue,
  notefont=\sffamily\bfseries,
  notebraces={(}{)},
  headpunct={},
  bodyfont=\normalfont,
  postheadspace=\newline%
]{myMainStyle}
\declaretheorem[style=myMainStyle, numberwithin=section]{definition}%
\declaretheorem[style=myMainStyle, sibling=definition]{proposition}

\declaretheorem[style=myMainStyle, sibling=definition]{remark}
\declaretheorem[style=myMainStyle, sibling=definition]{example}

\makeatletter
\preto\itemize{%
  \if@inlabel
    \ifstarttheorem
      \mbox{}\par\nobreak\vskip\glueexpr-\parskip-\baselineskip+0.25em\relax\hrule\@height\z@
    \fi%
  \fi%
  \global\starttheoremfalse%
 \def\tempa{proof}%
 \ifx\tempa\mycurrenvir
    \ifstarttheorem
      \mbox{}\par\nobreak\vskip\glueexpr-\parskip-\baselineskip+0.25em\relax\hrule\@height\z@
    \fi%
 \fi%
 \global\starttheoremfalse%
}
\preto\enditemize{\global\starttheoremfalse}
\makeatother

\makeatletter
\preto\enumerate{%
  \if@inlabel
    \ifstarttheorem
      \mbox{}\par\nobreak\vskip\glueexpr-\parskip-\baselineskip+0.25em\relax\hrule\@height\z@
    \fi%
  \fi%
  \global\starttheoremfalse%
 \def\tempa{proof}%
 \ifx\tempa\mycurrenvir
    \ifstarttheorem
      \mbox{}\par\nobreak\vskip\glueexpr-\parskip-\baselineskip+0.25em\relax\hrule\@height\z@
    \fi%
 \fi%
 \global\starttheoremfalse%
}
\preto\endenumerate{\global\starttheoremfalse}
\makeatother

\newcommand{\tmb}[3]{\underset{{\scriptscriptstyle #3}}{\overset{{\scriptscriptstyle #1}}{#2}}}%

\newcommand*{\isim}{\tmb{\text{\tiny{ind.}}}{\sim}{}}

\newcommand*{\IN}{\mathbb{N}}

\newcommand*{\U}{\operatorname{U}}

\newcommand*{\expm}{\operatorname*{expm1}}
\newcommand*{\logp}{\operatorname*{log1p}}

\renewcommand*{\P}{\mathbb{P}}

\newcommand*{\R}{\textsf{R}}

\IfFileExists{upquote.sty}{\usepackage{upquote}}{}
\begin{document}
\thispagestyle{plain}
\begin{center}
  \sffamily
  {\bfseries\LARGE Random number generators produce collisions:\\ Why, how many and more\par}
  \bigskip\smallskip
  {\Large
    Marius Hofert\footnote{Department of Statistics and Actuarial Science, University of
    Waterloo, 200 University Avenue West, Waterloo, ON, N2L
    3G1,
    \href{mailto:marius.hofert@uwaterloo.ca}{\nolinkurl{marius.hofert@uwaterloo.ca}}. This
    work was supported by NSERC under Discovery Grant RGPIN-5010-2015.}
    \par
    \bigskip
    \myisodate\par}
\end{center}
\par\smallskip
\begin{abstract}
  It seems surprising that when applying widely used random number generators to
  generate one million random numbers on modern architectures, one obtains, on
  average, about 116 collisions. This article explains why, how to
  mathematically compute such a number, why they often cannot be obtained in a
  straightforward way, how to numerically compute them in a robust way and,
  among other things, what would need to be changed to bring this number below
  1. The probability of at least one collision is also briefly addressed, which,
  as it turns out, again needs a careful numerical treatment. Overall, the
  article provides an introduction to the representation of floating-point
  numbers on a computer and corresponding implications in statistics and
  simulation. All computations are carried out in \R\ and are reproducible with
  the code included in this article.
\end{abstract}
\minisec{Keywords}
Random numbers, %
ties, collisions, floating point numbers, expected number of collisions, probability of collision.
\minisec{MSC2010}
65C60%

\section{Introduction}
When generating one million (pseudo-)random numbers %
(so realizations of independent and identically distributed (iid) random variables $U_1,\dots,U_n$
from the standard uniform distribution $\U(0,1)$, in short $U_1,\dots,U_n\isim\U(0,1)$ for $n=10^6$)
with the default random number generator (RNG) of the statistical software \R, see \cite{R}, one typically obtains
well over a hundred duplicated values. Let us first verify this fact empirically.
\begin{Schunk}
\begin{Sinput}
> n <- 1e6 # number of random numbers to draw
> set.seed(271) # set a seed for reproducibility
> U <- runif(n) # generate n U(0,1) realizations
> I <- duplicated(U) # logical(n) indicating duplicated values
> (C <- sum(I)) # number of duplicated values
\end{Sinput}
\begin{Soutput}
[1] 120
\end{Soutput}
\begin{Sinput}
> stopifnot(C == n - length(unique(U))) # sanity check
\end{Sinput}
\end{Schunk}
For the seed $271$ used here, \R's default RNG produced $120$ duplicated
values. It seems unexpected that this number is so large given that generating
$n=10^6$ random numbers is far away from being a large number in simulation
studies in this day and age, and given that any number of iid $\U(0,1)$ random
variables should, almost surely (so with probability $1$), not produce any
duplicated values. Why do we see so many of them? How many can we expect to see when
generating $n$ iid realizations from $\U(0,1)$?

Before finding answers to these questions, let us clarify some terms. If equal
numbers appear, the statistics literature typically speaks of \emph{tied data}
or \emph{ties} appearing in the data. It seems less clear whether two equal numbers are counted as one tie or two
ties. \cite[Example~1]{kendall1945} adopts the latter approach.
The following code computes the number of ties according to this convention in simple examples and for the realizations in
\code{U}. %
\begin{Schunk}
\begin{Sinput}
> ## Auxiliary function for indicating ties
> ties <- function(x) {
+     tab <- table(x) # count for each unique value of x how often it appears
+     tab[tab == 1] <- 0 # those appearing precisely once are no ties
+     tab
+ }
> ## Example 1
> x <- c(1, 2, 3)
> (dpl <- duplicated(x)) # x's duplicated?
\end{Sinput}
\begin{Soutput}
[1] FALSE FALSE FALSE
\end{Soutput}
\begin{Sinput}
> sum(dpl) # no duplicated values; note that R treats FALSE as 0 and TRUE as 1
\end{Sinput}
\begin{Soutput}
[1] 0
\end{Soutput}
\begin{Sinput}
> (tie <- ties(x)) # tied?
\end{Sinput}
\begin{Soutput}
x
1 2 3
0 0 0
\end{Soutput}
\begin{Sinput}
> sum(tie) # no ties
\end{Sinput}
\begin{Soutput}
[1] 0
\end{Soutput}
\begin{Sinput}
> ## Example 2
> x <- c(1, 1, 2)
> (dpl <- duplicated(x)) # x's duplicated?
\end{Sinput}
\begin{Soutput}
[1] FALSE  TRUE FALSE
\end{Soutput}
\begin{Sinput}
> sum(dpl) # one duplicated value (the second 1)
\end{Sinput}
\begin{Soutput}
[1] 1
\end{Soutput}
\begin{Sinput}
> (tie <- ties(x)) # tied?
\end{Sinput}
\begin{Soutput}
x
1 2
2 0
\end{Soutput}
\begin{Sinput}
> sum(tie) # two ties
\end{Sinput}
\begin{Soutput}
[1] 2
\end{Soutput}
\begin{Sinput}
> ## Example 3
> x <- c(1, 1, 2, 3, 3, 3)
> (dpl <- duplicated(x)) # x's duplicated?
\end{Sinput}
\begin{Soutput}
[1] FALSE  TRUE FALSE FALSE  TRUE  TRUE
\end{Soutput}
\begin{Sinput}
> sum(dpl) # three duplicated values (the second 1 and the second and third 3)
\end{Sinput}
\begin{Soutput}
[1] 3
\end{Soutput}
\begin{Sinput}
> (tie <- ties(x)) # tied?
\end{Sinput}
\begin{Soutput}
x
1 2 3
2 0 3
\end{Soutput}
\begin{Sinput}
> sum(tie) # five ties (both 1s and all three 3s)
\end{Sinput}
\begin{Soutput}
[1] 5
\end{Soutput}
\begin{Sinput}
> ## For our example above
> sum(duplicated(U)) # 120 duplicated values
\end{Sinput}
\begin{Soutput}
[1] 120
\end{Soutput}
\begin{Sinput}
> sum(ties(U)) # 240 ties
\end{Sinput}
\begin{Soutput}
[1] 240
\end{Soutput}
\end{Schunk}
We see that \code{duplicated()} indeed indicates colliding numbers or \emph{collisions} and
\code{sum(duplicated(U))} thus counts the number of collisions.
$C$ collisions correspond to at least $C+1$ ties (if all ties are equal) and at most $2C$ ties
(if there are $C$ different pairs of ties). The latter is the case for the $n=10^6$ random
numbers generated above, so for the chosen seed we obtained 120 collisions and 240 ties;
at least not more than two of any of the $n$ values generated are equal.
In what follows we focus on the number of collisions. %

\section{Basic analysis}
As a first step, we check how the collisions are distributed among the $n$
generated random numbers, in other words, which random numbers are tied.
We plot the index of each collision in the vector
$\bm{U}=(U_1,\dots,U_n)$ (so the position of each duplicated number within the $n$
generated random numbers) against the index of each collision (so $1$ to
$120$). By building cumulative sums with \code{cumsum()}, we can also easily plot the
number of collisions as a function of the sample size up to and including $n$.
Figure~\ref{fig:coll:distribution:num} shows these two plots.
\begin{Schunk}
\begin{Sinput}
> tab <- table(U) # count for each unique number how often it appears
> plot(which(tab > 1), xlab = "Index of each collision (duplicated number)",
+      ylab = "Index of each collision (random number appearing more than once)")
> plot(cumsum(I), type = "l", xlab = "Sample size",
+      ylab = "Number of collisions")
\end{Sinput}
\end{Schunk}
\setkeys{Gin}{width=\textwidth}
\begin{figure}[htbp]
\centering
  \begin{minipage}{0.48\textwidth}
\begin{Schunk}

\includegraphics[width=\maxwidth]{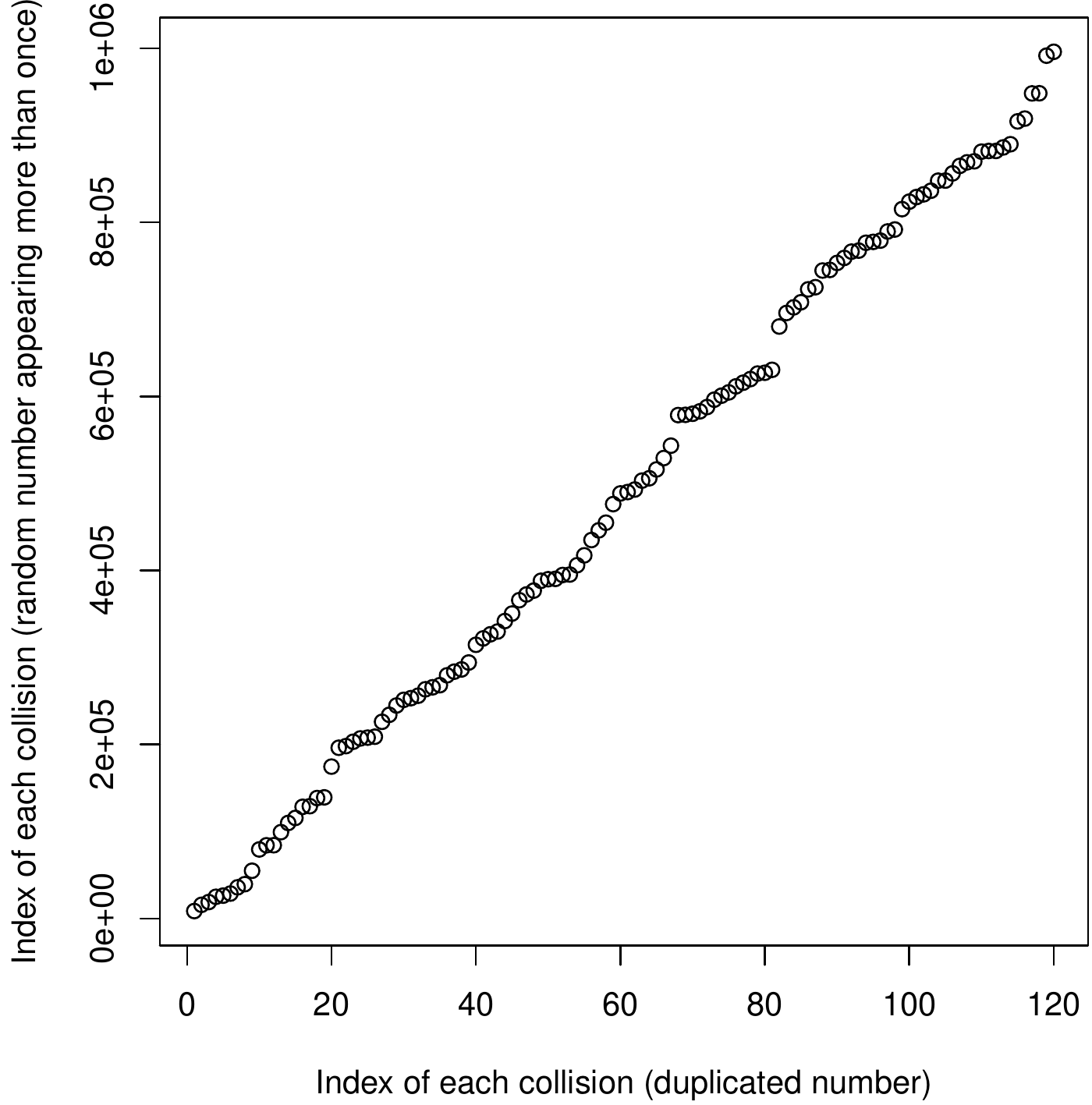} \end{Schunk}
  \end{minipage}%
  \hfill
  \begin{minipage}{0.48\textwidth}
\begin{Schunk}

\includegraphics[width=\maxwidth]{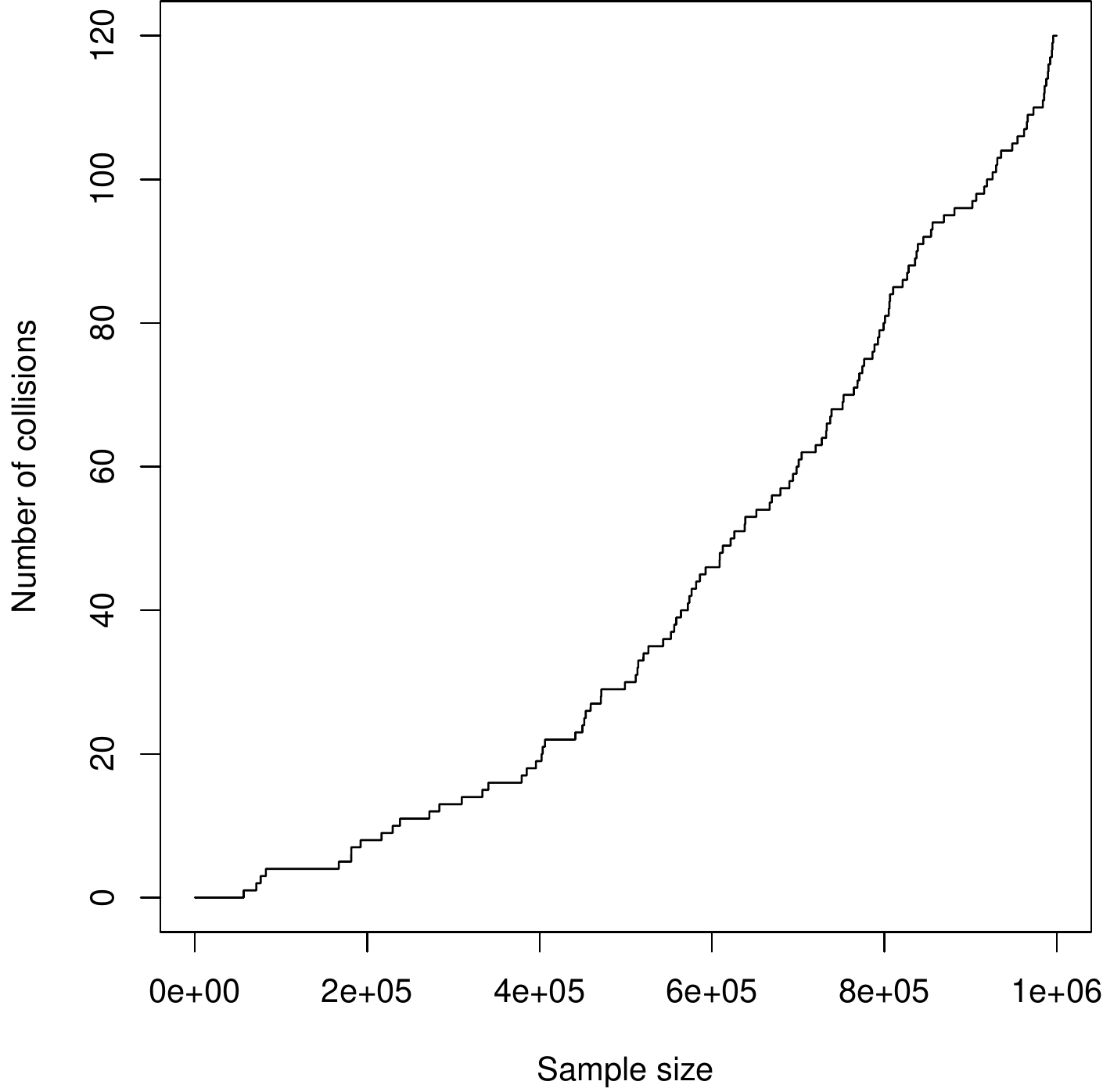} \end{Schunk}
  \end{minipage}
  \caption{Distribution of collisions (left) and number of collisions as a function of sample size (right).}
\label{fig:coll:distribution:num}
\end{figure}

As we can see from Figure~\ref{fig:coll:distribution:num}, the collisions do not show
any specific pattern and the number of collisions grows continuously as a
function of the sample size $n$.  In short, there does not seem to be any
strange behavior suddenly appearing. Maybe that is just the behavior
of \R's default RNG?

By default, \R\ uses the Mersenne Twister of \cite{matsumotonishimura1998} (with
initialization fixed by Professor Brian D.\ Ripley) for generating $\U(0,1)$
random numbers as can be seen from the first entry of the output of the \R\
function \code{RNGkind()}. Another popular RNG is L'Ecuyer's combined multiple
recursive RNG (CMRG); see \cite{lecuyer1999a} and
\cite{lecuyersimardchenkelton2002}.  We can use \code{RNGkind("L'Ecuyer-CMRG")}
to ask \R\ to change to this RNG, generate $n$ random numbers and count
the number of collisions.  We obtain $126$ for the seed $271$. So the
problem is not specific to the Mersenne Twister (and we thus switch back to it).
\begin{Schunk}
\begin{Sinput}
> RNGkind() # first component => Mersenne Twister is used as RNG for U(0,1)
\end{Sinput}
\begin{Soutput}
[1] "Mersenne-Twister" "Inversion"        "Rejection"
\end{Soutput}
\begin{Sinput}
> RNGkind("L'Ecuyer-CMRG") # switch random number generator (RNG)
> set.seed(271) # set a seed for reproducibility
> U. <- runif(n) # generate n U(0,1) realizations with L'Ecuyer's CMRG
> sum(duplicated(U.)) # count the number of collisions
\end{Sinput}
\begin{Soutput}
[1] 126
\end{Soutput}
\begin{Sinput}
> RNGkind("Mersenne-Twister") # switch back to R's default (Mersenne Twister)
\end{Sinput}
\end{Schunk}

\section{Floating-point numbers}\label{sec:floats}
The random numbers produced in \R\ are floating-point numbers. We now explain
what this means and the corresponding implications. It will be a key to
understand why RNGs produce collisions and how to (properly) compute the
expected number of collisions or the probability of a collision.

According to the IEEE~754 standard (\emph{IEEE} standing for the \emph{Institute of
Electrical and Electronics Engineers}), a 64-bit base-2 double-precision floating
point normal number (short: \emph{double}) is given by
\begin{align}
  x=(-1)^{s}(1.f_{51}\dots f_{0})_2\cdot 2^{e-1023}=(-1)^{s}\biggl(1+\sum_{i=1}^{52}f_{52-i}2^{-i}\biggr)\cdot 2^{e-1023}\label{double:repr}
\end{align}
which involves the following quantities:
\begin{itemize}
\item The \emph{sign} $s$. This single bit indicates the sign of $x$, so
  $s\in\{0,1\}$ with $s=1$ for negative $x$ and $s=0$ for positive $x$.
\item The \emph{significand} $f$. This group of $52$ bits represents
  the digits after the decimal point in ``$1.$'', so $f_i\in\{0,1\}$, $i=0,\dots,51$ (traditionally,
  the index $i$ starts from $0$). All numbers with significand of the form ``$1.$''
  are \emph{normal} numbers and for such numbers, all $52$ bits can be used
  to represent the digits after the decimal point. Note that the subscript $2$ in representation~\eqref{double:repr} indicates
  that the resulting number is to be viewed as represented in base 2.

  The significand is also known as \emph{mantissa} but the use of the latter
  term is discouraged by the IEEE or \cite{knuth1969} since ``mantissa'' has
  historically been used to refer to the factional part of a
  logarithm. Furthermore, note that ``normal number'' has a different meaning in
  computing (see above) than in mathematics (where a number is a \emph{normal
    number} if, for every $n\in\IN$, all blocks of $n$ digits are equally likely
  to appear). The meaning of a normal number in
  computing matches that of a \emph{normalized number} in mathematics.
\item The \emph{exponent} $e$. This group of $11$ bits represents the
  $2^{11}=2048$ numbers from $0$ to $2047$. The smallest exponent
  $e=(0\dots0)_2=0$ is used, for example to represent \emph{subnormal} numbers, so
  numbers whose significand does not begin with ``$1.$''. The largest exponent
  $e=(1\dots1)_2=1\cdot 2^0+\dots+1\cdot
  2^{10}=\sum_{i=0}^{10}2^i=\frac{1-2^{11}}{1-2}=2047$ is also reserved, for
  example for representing \code{Inf} (infinity) or \code{NaN} (not a
  number). Important for us is that the exponents from $e=(0\dots01)_2$ to
  $e=(1\dots10)_2$ (representing the numbers from $1$ to $2046$) can be used to
  represent finite, normal numbers. The shift of $e$ by $1023$ (which equals $2046/2$)
  in~\eqref{double:repr} allows for the representation of small or large
  exponents and therefore small or large floating point numbers $x$.
\end{itemize}
Figure~\ref{fig:IEEE:754} highlights these three components of a double.
\begin{figure}[htbp]
  \includegraphics[width=\textwidth]{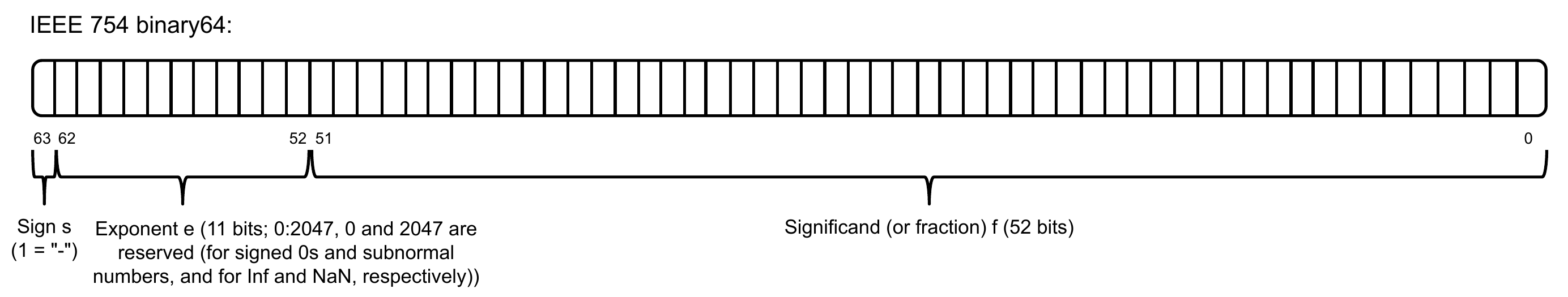}%
  \caption{Groups of bits (sign, exponent, significand) for representing floating point
    normal numbers~\eqref{double:repr} according to the IEEE~754 standard.}
  \label{fig:IEEE:754}
\end{figure}

To better understand these technicalities, we now consider some examples.
\begin{example}[Smallest positive normal number]\label{ex:smallest:positive}
  To obtain the smallest positive normal number $x$ of form~\eqref{double:repr}, we have to
  choose $s=0$ (since $x>0$), the smallest non-reserved exponent
  $e=(0\dots 01)_2=1$ and the smallest significand $f=(0\dots0)_2$. This leads
  to
  $x=(-1)^0\cdot (1.0)_2\cdot 2^{1-1023}= 1.0\cdot 2^{1-1023} = 2^{-1022}\approx 2.225074\cdot 10^{-308}$.
  \R\ stores this value as component \code{double.xmin} of the list of
  machine constants, \code{.Machine}.
\begin{Schunk}
\begin{Sinput}
> .Machine$double.xmin
\end{Sinput}
\begin{Soutput}
[1] 2.225074e-308
\end{Soutput}
\end{Schunk}
Since \code{.Machine$double.xmin} is the smallest positive normal number, any smaller number
should be truncated to zero, right? Not quite. As the
exponent $e=(0\dots0)_2$ is used to represent subnormal numbers, \R\ utilizes those
to obtain even smaller (subnormal) numbers and we can say exactly when truncation to zero happens:
\begin{Schunk}
\begin{Sinput}
> .Machine$double.xmin / 2 # still not 0 (but subnormal)
\end{Sinput}
\begin{Soutput}
[1] 1.112537e-308
\end{Soutput}
\begin{Sinput}
> .Machine$double.xmin/2^52 == 0 # significand (0..01)_2 after "0.", still not 0
\end{Sinput}
\begin{Soutput}
[1] FALSE
\end{Soutput}
\begin{Sinput}
> .Machine$double.xmin/2^53 == 0 # significand (0...0)_2 after "0.", now really 0
\end{Sinput}
\begin{Soutput}
[1] TRUE
\end{Soutput}
\end{Schunk}
\end{example}

\begin{example}[Smallest normal number greater than one minus 1]\label{ex:non:equidistant}
  To obtain the smallest normal number $x$ greater than
  one minus $1$ of form~\eqref{double:repr}, we have to choose $s=0$, the exponent $e=(01\dots
  1)_2=1023$ (to obtain the zero power of
  $2$ to be able to represent numbers
  $1.\dots$) and the smallest significand greater than $0$,
  $f=(0\dots01)_2$. This leads to
  $x=(-1)^0\cdot (1.0\dots01)_2 \cdot 2^{1023-1023} -1 = 1+2^{-52} -1 = 2^{-52} \approx 2.220446\cdot 10^{-16}$.
  \R\ stores this value as component \code{double.eps} of \code{.Machine}.
\begin{Schunk}
\begin{Sinput}
> .Machine$double.eps
\end{Sinput}
\begin{Soutput}
[1] 2.220446e-16
\end{Soutput}
\end{Schunk}
\end{example}

\begin{example}[Largest normal number]
  To obtain the largest normal number $x$ of form~\eqref{double:repr}, we have to choose $s=0$, the largest non-reserved exponent $e=(1\dots
  10)_2=\sum_{i=1}^{10}2^i=2\sum_{i=0}^9 2^i=2\frac{1-2^{10}}{1-2}=2^{11}-2=2046$ and the largest significand
  $f=(1\dots1)_2$. This leads to
  $x=(-1)^0\cdot (1.1\dots1)_2 \cdot 2^{2046-1023} = (1+\sum_{i=1}^{52}2^{-i})2^{1023} = (\sum_{i=0}^{52} 2^{-i})2^{1023} = \bigl(\frac{1-(1/2)^{53}}{1-1/2}\bigr)2^{1023} =(1-2^{-53})2^{1024}$. Note that this cannot be evaluated as such since \code{2^1024} is \code{Inf} in double precision. But if we rewrite it as $(1-2^{-53})2^{1024}=(2^{53}-1)2^{971}$ we obtain $\approx 1.797693\cdot 10^{308}$.
  \R\ stores this value as component \code{double.xmax} of \code{.Machine}.
\begin{Schunk}
\begin{Sinput}
> .Machine$double.xmax
\end{Sinput}
\begin{Soutput}
[1] 1.797693e+308
\end{Soutput}
\end{Schunk}
\end{example}

\begin{remark}\label{rem:rng}
  Since \R\ version 3.6.0, the default algorithm for generating random integers
  from the discrete uniform distribution
  $\U(\{1,\dots,n\})$ on the numbers 1 to
  $n$ is rejection; see the third entry \code{"Rejection"} that \code{RNGkind()} printed above and its
  interpretation on \code{?RNGkind}.  The idea is that at least
  $k=\lceil\log_2(n)\rceil$-many bits are required to represent the numbers 0 to
  $n-1$ (but potentially more) and so one can randomly generate
  $k$-bit patterns until a number less than or equal to
  $n-1$ results; shifting the output by 1 then leads to realizations from 1 to
  $n$ as required for a $\U(\{1,\dots,n\})$ generator.

  As such, one might be tempted to believe that one could generate $\U(0,1)$
  numbers by generating random 52-bit patterns for the significand.  However,
  the resulting doubles are not equally spaced, they are denser near 0 than near 1
  which we have seen in Examples~\ref{ex:smallest:positive} and~\ref{ex:non:equidistant}:
  The distance between the smallest positive normal number and $0$ is the smallest positive normal number, so
of the order $\mathcal{O}(10^{-308})$ whereas the distance between the smallest normal number greater than $1$ and $1$
is of the order $\mathcal{O}(10^{-16})$, so much larger. In particular, the grid of representable doubles is not equidistant. This is why numerical problems often appear near $1$ and
why \R\ provides numerically stable functions such as \code{expm1()} (for $\expm(x)=\exp(x)-1$) and \code{log1p()} (for $\logp(x)=\log(1+x)$)
to evaluate $\exp(x)-1$ and $\log(1+x)$ for arguments $x$ near $0$; see the appendix for more details.
\end{remark}

\section{Theoretical results}
Many RNGs produce (pseudo-)random integers in the form of a
recursion and such integers are then location-scale transformed to $(0,1)$ and
represented as doubles to produce (pseudo-)random numbers. If one consults the
help page of \code{RNGkind()} via \code{?RNGkind}, one can learn that the
recursions of both the Mersenne Twister and L'Ecuyer's CMRG produce 32-bit
integers. The following result provides the expected number of collisions in a
$k$-bit integer architecture (in which we can represent the $2^k$ integers from
$0$ to $2^k - 1$) with a RNG that produces $k$-bit integers which are then mapped
to doubles in $(0,1)$ and returned as (pseudo-)random numbers.
\begin{proposition}[Expected number of collisions]\label{prop:exp:num:coll}
  The expected number of collisions among $n$ random $k$-bit integers is
  $n - 2^k (1 - (1 - 2^{-k})^n)$.
\end{proposition}
\begin{proof}
  Let $i$ be a fixed $k$-bit integer and $I_1,\dots,I_n$ be $n$ random $k$-bit
  integers. The probability $p$ of $i$ not being attained by a single $I_l$ is
  $p=\P(I_l\neq i)=(2^k-1) / 2^k = 1-2^{-k}$ ($I_l$ can take on $2^k-1$ out of $2^k$
  values). Therefore, the probability of none of $I_1,\dots,I_n$ to attain $i$ is
  $\P(I_1\neq i,\dots,I_n\neq i)=p^n$ and the probability of $i$ being attained
  by (at least) one $I_l$ is
  $\P(I_l=i\ \text{for at least one}\ l=1,\dots,n)=1-p^n$. The expected number
  of distinct $i$'s among the $2^k$-many to be attained by at least one $I_l$ is thus
  $2^k (1 - p^n)$. Therefore, the expected number of collisions among $n$ random
  $k$-bit integers is $n$ minus the expected number of distinct, attained $i$'s,
  so $n - 2^k (1 - p^n) = n - 2^k (1 - (1 - 2^{-k})^n)$.
\end{proof}

\begin{remark}\label{rem:theory}
  \begin{enumerate}
  \item The expected number of collisions given in Proposition~\ref{prop:exp:num:coll}
    is only correct under the stated assumption that indeed random $k$-bit
    integers are generated. On a $k$-bit architecture, this is only correct if all
    $k$ bits are used. In what follows we assume to work with a $k$-bit
    architecture and an RNG that produces $k$-bit integers which are then mapped
    to doubles in $(0,1)$ and returned as (pseudo-)random numbers. We simply refer
    to such as setup as a \emph{$k$-bit setup}. Note that not all RNGs produce a
    block of $k$ random bits on a $k$-bit architecture;
    see \cite{lecuyer1999a}, \cite{lecuyer1999b}, \cite{lecuyer2017},
    \cite{lecuyermungeroreshkinsimard2017}, \cite{lecuyersimard2007} and the
    references therein.
  \item\label{rem:theory:2}
    \cite{lecuyersimardwegenkittl2002} study the exact distribution of the
    number $C$ of collisions and derive %
    that $\P(C = c)=\frac{b(b-1)\cdot\dots\cdot(b-n+1+c)}{b^n}S(n,n-c)$, %
    where $b$ is the number of buckets (here: $b=2^k$ for $k$-bit integers) the $n$ random numbers can fall into
    and where $S(n,l)$ is the Stirling number of the second kind (providing
    the number of ways $n$ objects can be partitioned into $l$ non-empty sets).
    If $n>b$, the \emph{pigeonhole principle} implies that there must be a collision
    and if $n\le b$, the probability of at least one collision is $\P(C\ge 1)=1-\P(C=0)=1-\frac{(b)_n}{b^n}\cdot 1=1-\prod_{i=1}^{n-1}(1-\frac{i}{b})$
    which is well-known from the \emph{birthday problem}; see, for example, \cite{diaconismosteller1989}. %
    This probability can be computed in \R\ with \code{pbirthday(n, classes = b)}, %
    where \code{n} and \code{b} denote $n$ and $b$, respectively. We will come back to \code{pbirthday()} later.
  \end{enumerate}
\end{remark}

Let us now compute the expected number of collisions for $n=10^6$ and $k=32$
according to Proposition~\ref{prop:exp:num:coll}.
\begin{Schunk}
\begin{Sinput}
> k <- 32 # number of bits of the architecture
> n - 2^k * (1 - (1 - 2^{-k})^n) # expected # of collisions for n random numbers
\end{Sinput}
\begin{Soutput}
[1] 116.4062
\end{Soutput}
\end{Schunk}
So for our original problem of $10^6$ random numbers generated under a 32-bit setup, we indeed expect about 116
collisions.

\section{Numerical evaluation for different setups and sample sizes}
Instead of computing single numbers, it is typically a good idea to consider the output as a
function of the input and check the output for multiple inputs if possible (or
at least select inputs, if necessary).  Our problem is simple enough that we can
study it as a function of $k$ (the number of bits of our setup). Or is it
not? Clearly, we would expect to see fewer collisions, the larger $k$. Let us check this.
\begin{Schunk}
\begin{Sinput}
> k <- 32:64
> plot(k, n - 2^k * (1 - (1 - 2^{-k})^n), type = "l", xlab = "Number k of bits",
+      ylab = paste("Expected number of collisions when generating", n,
+                   "random numbers"))
\end{Sinput}
\end{Schunk}
As we see from the left-hand side of Figure~\ref{fig:exp:num:coll}, there seems
to be a problem. The expected number of collisions appears to be an increasing function of $k$
with a sudden jump at $k=54$. Something is wrong.

Let us consider this problem more closely. We know that $n - 2^k (1 - (1 - 2^{-k})^n)$
is the mathematically correct answer, but is it also the numerically correct one? The inner power $2^{-k}$ looks fine:
\begin{Schunk}
\begin{Sinput}
> 2^(-k)
\end{Sinput}
\begin{Soutput}
 [1] 2.328306e-10 1.164153e-10 5.820766e-11 2.910383e-11 1.455192e-11
 [6] 7.275958e-12 3.637979e-12 1.818989e-12 9.094947e-13 4.547474e-13
[11] 2.273737e-13 1.136868e-13 5.684342e-14 2.842171e-14 1.421085e-14
[16] 7.105427e-15 3.552714e-15 1.776357e-15 8.881784e-16 4.440892e-16
[21] 2.220446e-16 1.110223e-16 5.551115e-17 2.775558e-17 1.387779e-17
[26] 6.938894e-18 3.469447e-18 1.734723e-18 8.673617e-19 4.336809e-19
[31] 2.168404e-19 1.084202e-19 5.421011e-20
\end{Soutput}
\end{Schunk}
The innermost difference shows:
\begin{Schunk}
\begin{Sinput}
> 1-2^(-k) # really 1 or just printed as 1?
\end{Sinput}
\begin{Soutput}
 [1] 1 1 1 1 1 1 1 1 1 1 1 1 1 1 1 1 1 1 1 1 1 1 1 1 1 1 1 1 1 1 1 1 1
\end{Soutput}
\begin{Sinput}
> 1-2^(-k) == 1 # 1-2^{-k} becomes indistinguishable from 1 although 2^{-k} != 0
\end{Sinput}
\begin{Soutput}
 [1] FALSE FALSE FALSE FALSE FALSE FALSE FALSE FALSE FALSE FALSE FALSE FALSE
[13] FALSE FALSE FALSE FALSE FALSE FALSE FALSE FALSE FALSE FALSE  TRUE  TRUE
[25]  TRUE  TRUE  TRUE  TRUE  TRUE  TRUE  TRUE  TRUE  TRUE
\end{Soutput}
\begin{Sinput}
> k[min(which(1-2^(-k) == 1))] # smallest k for which the problem appears
\end{Sinput}
\begin{Soutput}
[1] 54
\end{Soutput}
\end{Schunk}
We see that the part $1 - 2^{-k}$ in the formula for the expected number of collisions
becomes numerically indistinguishable from $1$ for $k\ge 54$ even though $2^{-k}\neq 0$ for all $k$ considered.
This problem appears because of what we mentioned in Remark~\ref{rem:rng}, floating point normal numbers
are denser near $0$ than near $1$. To detect this problem it was helpful to consider the plot
on the left-hand side of Figure~\ref{fig:exp:num:coll} instead of just considering a single $k$.

To correctly compute the expected number of collisions, we rewrite the part
$1-(1-2^{-k})^n$ to avoid the direct evaluation of $1-2^{-k}$. We have
\begin{align*}
  1-(1-2^{-k})^n &= 1 - \exp(n\log(1 - 2^{-k})) = 1 - \exp(n\logp(-2^{-k}))\\
  &= -\expm(n\log1p(-2^{-k})),
\end{align*}
with $\expm(x)=\exp(x)-1$ and $\logp(x)=\log(1+x)$ as already introduced before.
We can now write a function utilizing this formula.
\begin{Schunk}
\begin{Sinput}
> exp_num_coll <- function(n, k) n + 2^k * expm1(n * log1p(-2^(-k)))
\end{Sinput}
\end{Schunk}
If we then plot the expected number of collisions among $n$ random numbers in a $k$-bit setup
as a function of $k$, we indeed
obtain a decreasing graph; see the right-hand side of Figure~\ref{fig:exp:num:coll}.
\begin{Schunk}
\begin{Sinput}
> y <- exp_num_coll(n, k = k)
> plot(k, y, type = "l", xlab = "Number k of bits",
+      ylab = paste("Expected number of collisions when generating", n,
+                   "random numbers"))
\end{Sinput}
\end{Schunk}
\setkeys{Gin}{width=\textwidth}
\begin{figure}[htbp]
\centering
  \begin{minipage}{0.48\textwidth}
\begin{Schunk}

\includegraphics[width=\maxwidth]{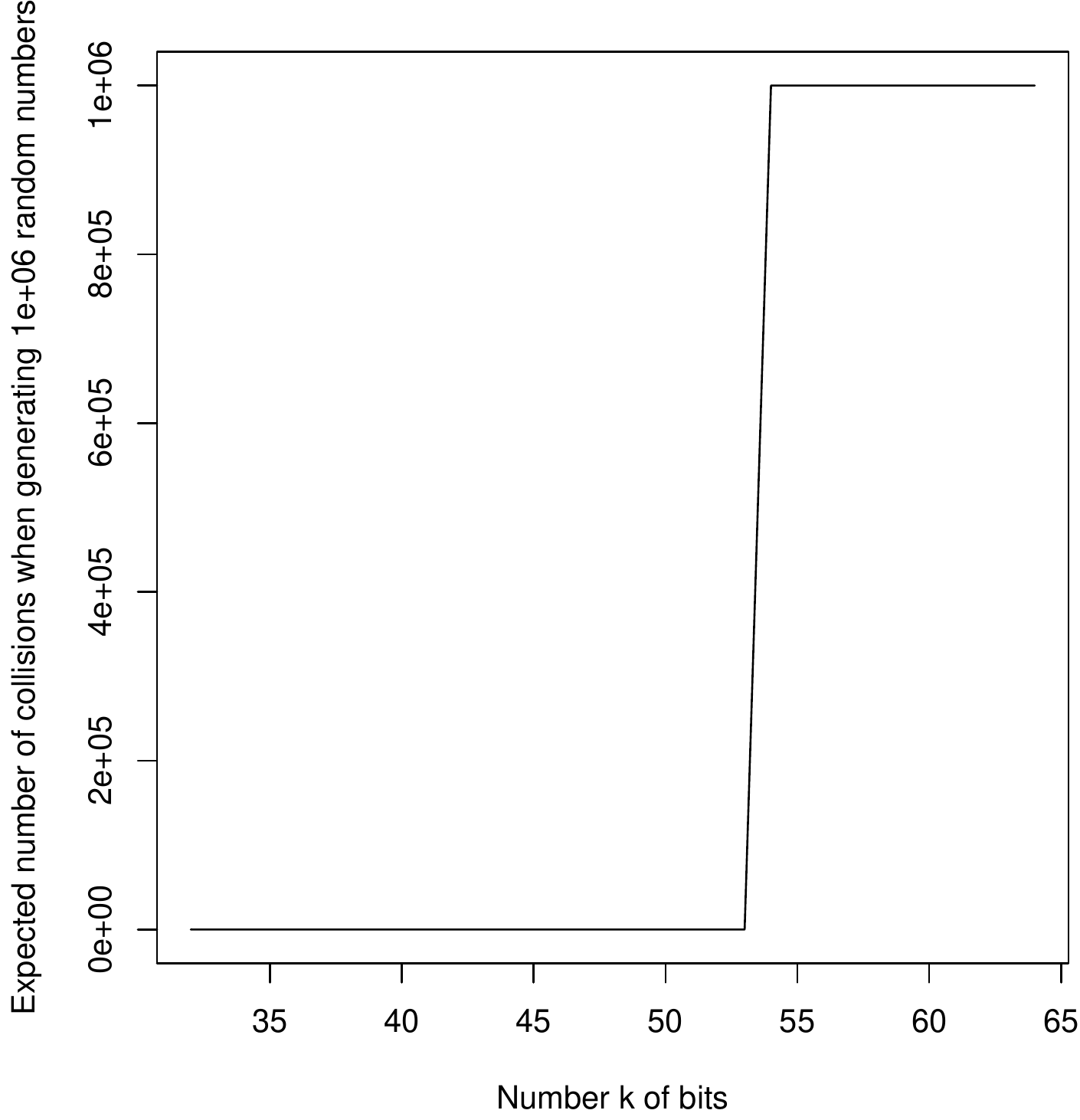} \end{Schunk}
  \end{minipage}%
  \hfill
  \begin{minipage}{0.48\textwidth}
\begin{Schunk}

\includegraphics[width=\maxwidth]{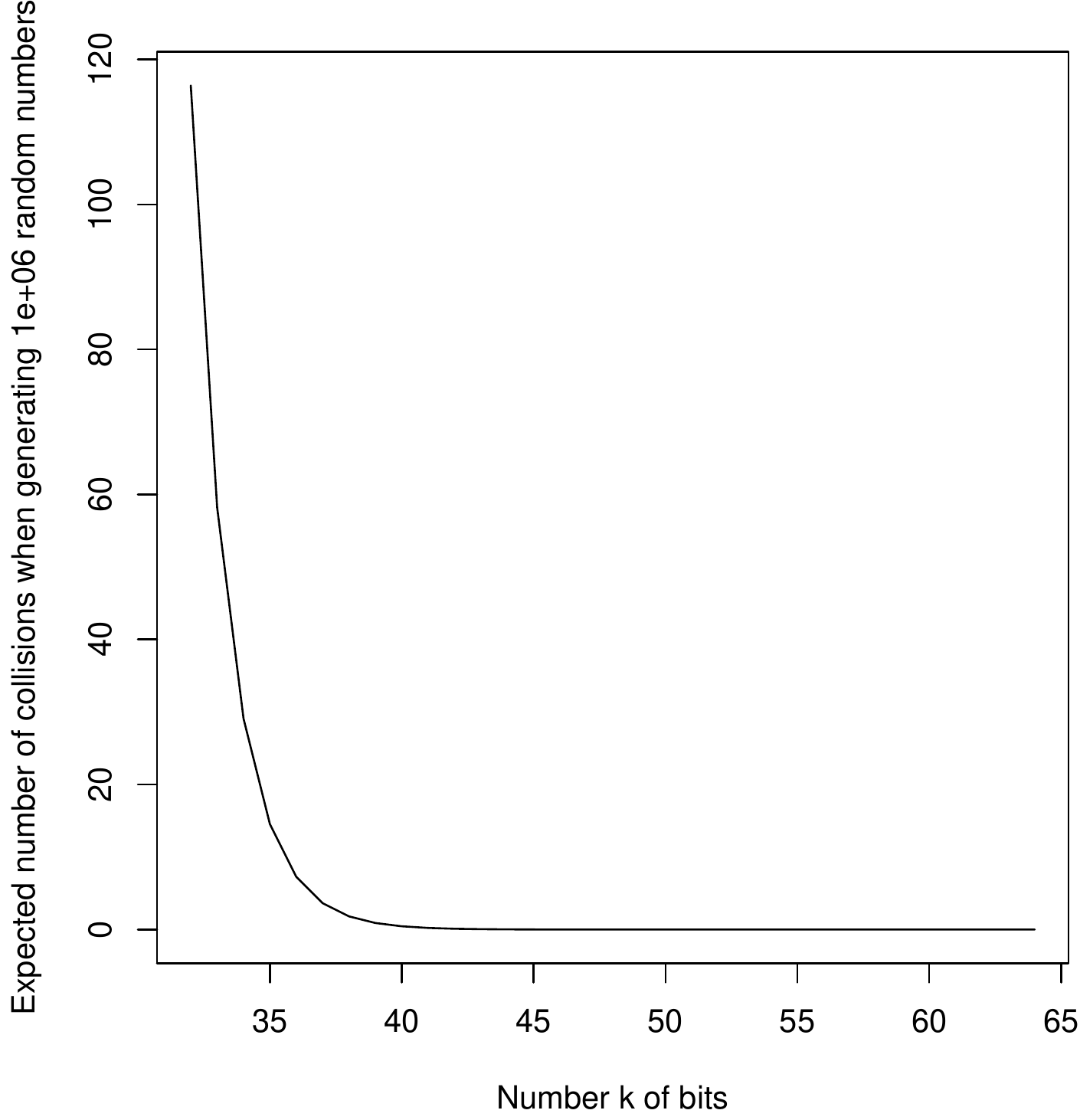} \end{Schunk}
  \end{minipage}
  \caption{Expected number of collisions among $n$ random numbers in a $k$-bit setup
    computed as $n - 2^k (1 - (1 - 2^{-k})^n)$ (left) and as $n + 2^k \expm(n\logp(-2^{-k}))$ (right).}
\label{fig:exp:num:coll}
\end{figure}

Now that we have a numerically robust formula, we can answer some more questions. What
is the smallest $k$ such that generating one million random numbers in a $k$-bit setup
leads to, in expectation, at most one collision?
\begin{Schunk}
\begin{Sinput}
> k[which(y <= 1)[1]]
\end{Sinput}
\begin{Soutput}
[1] 39
\end{Soutput}
\end{Schunk}
What is the expected number of collisions among one million random numbers in a 52-bit
setup?
\begin{Schunk}
\begin{Sinput}
> exp_num_coll(n, k = 52)
\end{Sinput}
\begin{Soutput}
[1] 0.0001110223
\end{Soutput}
\end{Schunk}

A realistic scenario is the shift from a 32-bit to a 64-bit setup at some point.
In what follows, we consider the presented formulas for $k$-bit setups up to
$k=64$, but note that if the generated random numbers are stored as doubles, the
expected number and the probability of collisions among the resulting random
numbers are most likely higher than those of $k$-bit integers. This is because
the significand only has 52 bits. It seems hard to determine the probability of
collision or the expected number of collisions exactly because of subnormal
numbers and the fact that doubles are not equally spaced in the unit interval;
see Remark~\ref{rem:rng}.

The expected number of collisions %
and the number $n$ %
under a 64-bit setup so that we can expect to
see one collision can be obtained as follows.
\begin{Schunk}
\begin{Sinput}
> exp_num_coll(n, k = 64)
\end{Sinput}
\begin{Soutput}
[1] 2.712477e-08
\end{Soutput}
\begin{Sinput}
> uniroot(function(n)
+     exp_num_coll(n, k = 64) - 1, lower = 1e6, upper = 1e10)[["root"]]
\end{Sinput}
\begin{Soutput}
[1] 6.074e+09
\end{Soutput}
\end{Schunk}

Coming back to the birthday problem, what is the probability of seeing at least one collision %
under a 64-bit setup?
\begin{Schunk}
\begin{Sinput}
> pbirthday(n, classes = 2^64)
\end{Sinput}
\begin{Soutput}
Warning in (c:(c - n + 1))/rep(c, n): longer object length is not a multiple of shorter object length
\end{Soutput}
\begin{Soutput}
[1] 2.707387e-08
\end{Soutput}
\end{Schunk}
Looking at the rather short code of \code{pbirthday()} (\R\ version 4.0.0), we see that
the above call simply returns \code{1 - prod((c:(c - n + 1))/rep(c, n))}, where \code{c}
equals \code{classes}, so $2^{64}$. This is a straightforward evaluation of the probability
of at least one collision according to the birthday problem; see Remark~\ref{rem:theory}~\ref{rem:theory:2}.
The same warning appears already for much smaller $n$ and smaller $k$. %
\begin{Schunk}
\begin{Sinput}
> pbirthday(3, classes = 2^53) # fine
\end{Sinput}
\begin{Soutput}
[1] 3.330669e-16
\end{Soutput}
\begin{Sinput}
> pbirthday(3, classes = 2^54) # warning
\end{Sinput}
\begin{Soutput}
Warning in (c:(c - n + 1))/rep(c, n): longer object length is not a multiple of shorter object length
\end{Soutput}
\begin{Soutput}
[1] 5.551115e-16
\end{Soutput}
\begin{Sinput}
> length((2^54):((2^54)-3+1)) # not of length 3 anymore
\end{Sinput}
\begin{Soutput}
[1] 5
\end{Soutput}
\end{Schunk}
The warning comes from the fact that \code{c:(c - n + 1)} is not a sequence of length \code{n}
anymore for large number of classes. More severe problems start to appear for larger number of classes.
\begin{Schunk}
\begin{Sinput}
> stopifnot(pbirthday(3, classes = 2^56) == 0) # numerically even 0 (wrong)
\end{Sinput}
\end{Schunk}
For a numerically robust treatment of the probability of at least one collision, %
we can use a similar trick as before. If $b$ denotes the number of buckets (with $b=2^k$ here) as before,
then
\begin{align*}
  \P(C\ge 1)&=1-\prod_{i=1}^{n-1}\biggl(1-\frac{i}{b}\biggr)=1-\exp\biggl(\log\biggr(\prod_{i=1}^{n-1}\biggr(1-\frac{i}{b}\biggr)\biggr)\biggr) = 1-\exp\biggl(\,\sum_{i=1}^{n-1}\log\biggl(1-\frac{i}{b}\biggr)\biggr)\\
  &=1-\exp\biggl(\,\sum_{i=1}^{n-1}\logp\biggl(-\frac{i}{b}\biggr)\biggr)=-\expm\biggl(\,\sum_{i=1}^{n-1}\logp\biggl(-\frac{i}{b}\biggr)\Bigr).
\end{align*}
Using this, we can implement an improved version of \code{pbirthday()} in our case.
\begin{Schunk}
\begin{Sinput}
> pbirthday2 <- function(n, classes = 365)
+     if(n >= 2) {
+         if(n > classes) { # pigeonhole principle
+             1
+         } else { # numerically stable evaluation
+             -expm1(sum(log1p(-(1:(n-1))/classes)))
+         }
+     } else 0 # no collision possible if n <= 1
\end{Sinput}
\end{Schunk}
We now obtain:
\begin{Schunk}
\begin{Sinput}
> ## Simple check
> stopifnot(pbirthday2(3, classes = 2) == pbirthday(3, classes = 2)) # identical
> ## For 2^53 classes, we already only have numerical equality
> stopifnot(all.equal(pbirthday2(3, classes = 2^53),
+                     pbirthday (3, classes = 2^53))) # numerically equal
> ## The following provide different results than pbirthday() before
> pbirthday2(3, classes = 2^54) # fine; note how the value differs from before
\end{Sinput}
\begin{Soutput}
[1] 1.665335e-16
\end{Soutput}
\begin{Sinput}
> pbirthday2(n, classes = 2^64) # fine; note how the value differs from before
\end{Sinput}
\begin{Soutput}
[1] 2.710503e-08
\end{Soutput}
\end{Schunk}
Figure~\ref{fig:rel:err:pbirthday} shows the relative error when computing the probability of at least one collision with
\code{pbirthday()} in comparison to \code{pbirthday2()} for $n=10^6$ and $b=2^k$ for $k\in\{32,33,\dots,64\}$;
values equal to 0 for small such $k$ are omitted as the y-axis is plotted in logarithmic scale.
\begin{Schunk}
\begin{Sinput}
> y1 <- sapply(k, function(k.) pbirthday (n, classes = 2^k.))
> y2 <- sapply(k, function(k.) pbirthday2(n, classes = 2^k.))
> err <- abs(y2-y1) / y2 # relative error in comparison to y2
> err[err == 0] <- NA # for plotting in log-scale
> plot(k, err, type = "l", log = "y", xlab = "Number k of bits", ylab = paste(
+      "Relative error between probabilities of collision for sample size ", n))
\end{Sinput}
\end{Schunk}
\setkeys{Gin}{width=\textwidth}
\begin{figure}[htbp]
  \centering
  \begin{minipage}{0.48\textwidth}
\begin{Schunk}

\includegraphics[width=\maxwidth]{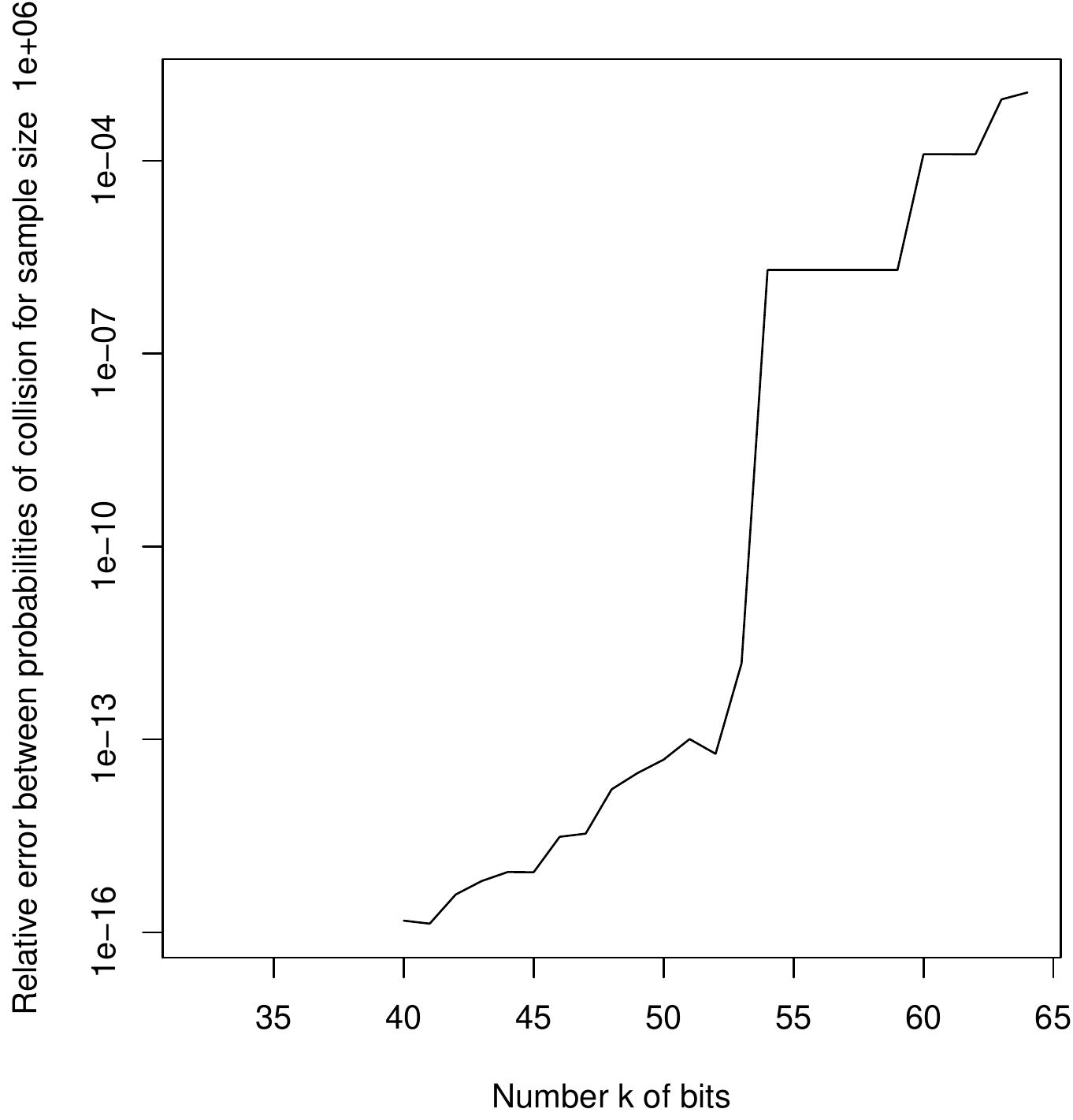} \end{Schunk}
  \end{minipage}%
  \caption{Relative error of the probability of at least one collision computed with
   \code{pbirthday()} in comparison to \code{pbirthday2()} for $n=10^6$ and $b=2^k$ for $k\in\{32,33,\dots,64\}$.}
  \label{fig:rel:err:pbirthday}
\end{figure}

\section{Conclusion}\label{sec:concl}
Rather unintuitively (but, as it turns out, not unexpectedly), generating
random numbers can lead to collisions. For generating $n$ random numbers in a
$k$-bit setup, the expected number of collisions is $n - 2^k (1
- (1 - 2^{-k})^n)$. We derived a numerically robust formula to
compute this number. For a 32-bit setup as is still used in random
number generators to this day (be it for historical reasons, reproducibility or due to run time),
the expected number of collisions when generating one million random numbers is about 116.
We also saw that for larger $k$, this number quickly becomes negligibly small.
Furthermore, we provided a numerically stable treatment of the probability of collision.

There are several important lessons to take away. First, RNGs are not perfect as
they are affected by the representation of numbers on a computer, typically in double precision.
Second, even when evaluating mathematically innocent formulas such as
$n - 2^k (1 - (1 - 2^{-k})^n)$ or $1-\prod_{i=1}^{n-1}(1-\frac{i}{b})$, one needs to keep an eye on numerical
issues. Third, the numerical behavior of mathematically equivalent functions may
be quite different; see Figure~\ref{fig:exp:num:coll}.  Fourth, numerical
issues are typically best detected when plotting a formula as a function of (at
least some of) its input variables. Last but not least, numerically robust
functions such as \code{expm1()} or \code{log1p()} are often helpful for obtaining
numerically robust evaluations of such formulas.

\subsection*{Acknowledgments}
I would like to thank Professor Ivan\ Kojadinovic (Universit\'e\ de\ Pau\ et\ des\ Pays\ de\ l'Adour)
and Professor Gianfausto Salvadori (Universit\`a del Salento) for motivating this problem. I would
also like to thank the Editor Daniel R.\ Jeske, an Associate Editor and two anonymous reviewers
for their insightful feedback.

\appendix

\section*{Appendix: How expm1() works}
The reader might wonder how, for example, \code{expm1()} actually works. In \R\ versions less than
3.5.0 (3.4.4 was the last official such release),
C~code of this function was found in \code{./src/nmath/expm1.c} of the \R\ sources
available on \url{https://cran.r-project.org/}; in later versions, \code{expm1()} directly stems
from the GNU Compiler Collection (see \url{https://gcc.gnu.org/}, then Git read access, then \code{./libgo/go/math/expm1.go})
and utilizes a more involved algorithm. %
Here is the C~code of \R\ version 3.4.4, where \code{DBL_EPSILON} has the meaning of the aforementioned \code{.Machine$double.eps}.
\begin{otherinput}
double expm1(double x)
{
    double y, a = fabs(x);

    if (a < DBL_EPSILON) return x;
    if (a > 0.697) return exp(x) - 1;  /* negligible cancellation */

    if (a > 1e-8)
	y = exp(x) - 1;
    else /* Taylor expansion, more accurate in this range */
	y = (x / 2 + 1) * x;

    /* Newton step for solving   log(1 + y) = x   for y : */
    /* WARNING: does not work for y ~ -1: bug in 1.5.0 */
    y -= (1 + y) * (log1p (y) - x);
    return y;
}
\end{otherinput}
For $|x|>0.697$, $\exp(x)-1$ is computed in a straightforward manner. For $|x|\in(10^{-8},\ 0.697]$,
also $\exp(x)-1$ is computed, but an additional Newton step added. Note that $\exp(x)-1=y$
if and only if $\log(1+y)-x = 0$ and the Newton step thus aims at improving the root of $f(y)=\log(1+y)-x$
via $y_{n+1}=y_n-f(y_n)/f'(y_n)=y-(1+y)(\logp(y)-x)$.
For $|x|\in[\text{\code{.Machine$double.eps}},\ 10^{-8}]$, on uses the approximation $\exp(x)-1\approx (x/2+1) x= x^2/2 + x$
and a Newton step as before. And for even smaller $x$, one simply uses $\exp(x)-1\approx x$. The left-hand side of
Figure~\ref{fig:expm1:log:exp} illustrates these different regions of approximation.
\begin{Schunk}
\begin{Sinput}
> m <- -26 # exponent of smallest x-value
> x <- 10^seq(m, 0) # sequence of x values
> y <- exp(x)-1 # numerically critical for |x| ~= 0
> y. <- expm1(x) # numerically stable
> plot(x, y., type = "l", log = "xy", ylab = "")
> lines(x, y, col = adjustcolor("black", alpha.f = 0.3), lwd = 6)
> legend("topleft", bty = "n", lty = c(1,1), lwd = c(4,1),
+        col = c(adjustcolor("black", alpha.f = 0.3), "black"),
+        legend = c("exp(x)-1", "expm1(x)"))
> x3 <- log(0.697, base = 10) # ~= -0.1567672 # threshold 3
> x2 <- -8 # threshold 2
> x1 <- log(.Machine$double.eps, base = 10) # threshold 1; ~= -15.65356
> abline(v = 10^c(x1, x2, x3), lty = 3) # vertical dotted lines
> ypos <- 10^(m + 2.5) # y-value for all labels
> text(3, ypos, labels = "exp(x)-1", srt = 90)
> text(10^mean(c(x2, x3)), ypos,
+      labels = "y = exp(x)-1\n+ Newton step\nfor solving f(y)\n=log(1+y)-x=0")
> text(10^mean(c(x1, x2)), ypos,
+      labels = "y = (1+x/2)x\n+ Newton step\nfor solving f(y)\n=log(1+y)-x=0")
> text(10^mean(c(m, x1)), ypos/9, labels = "x") # x
\end{Sinput}
\end{Schunk}

The right-hand side of Figure~\ref{fig:expm1:log:exp} shows a similar numerical issue
when computing $\log(\exp(x))$ for $x$ near $0$. This demonstrates that it is rarely a
good idea to compute the logarithm of a function $f$ as just $\log(f(x))$ and is indeed
 the reason why many densities in \R\ have an argument to allow for a numerically
 proper computation of their logarithms (which is of course needed for likelihood-based inference);
 see, for example, the argument \code{log} of \code{dnorm()}.
\begin{Schunk}
\begin{Sinput}
> plot(x, x, type = "l", log = "xy", ylab = "")
> lines(x, log(exp(x)), col = adjustcolor("black", alpha.f = 0.3), lwd = 6)
> legend("topleft", bty = "n", lty = c(1,1), lwd = c(4,1),
+        col = c(adjustcolor("black", alpha.f = 0.3), "black"),
+        legend = c("log(exp(x))", "x"))
\end{Sinput}
\end{Schunk}
\setkeys{Gin}{width=\textwidth}
\begin{figure}[htbp]
\centering
  \begin{minipage}{0.48\textwidth}
\begin{Schunk}

\includegraphics[width=\maxwidth]{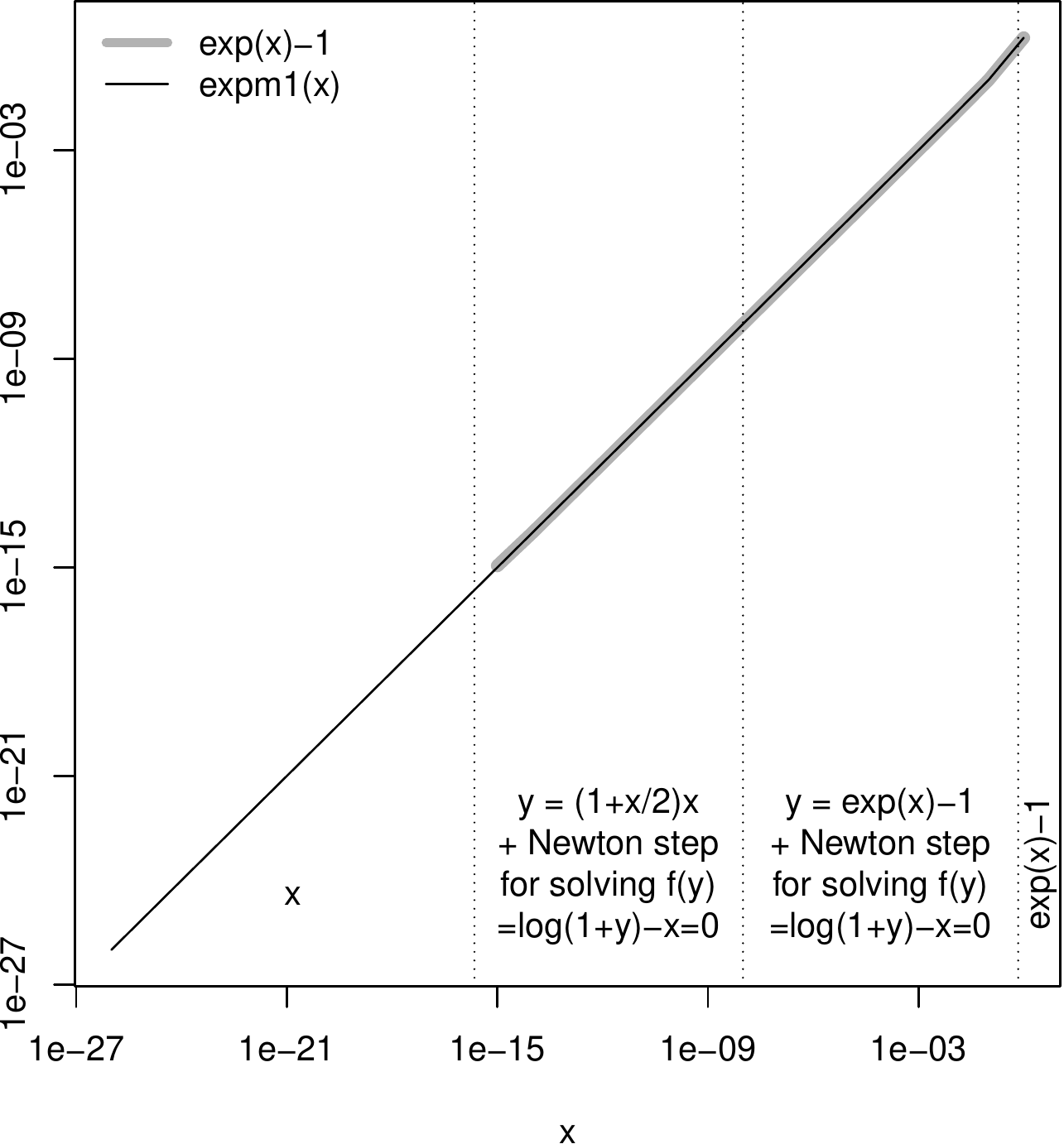} \end{Schunk}
  \end{minipage}%
  \hfill
  \begin{minipage}{0.48\textwidth}
\begin{Schunk}

\includegraphics[width=\maxwidth]{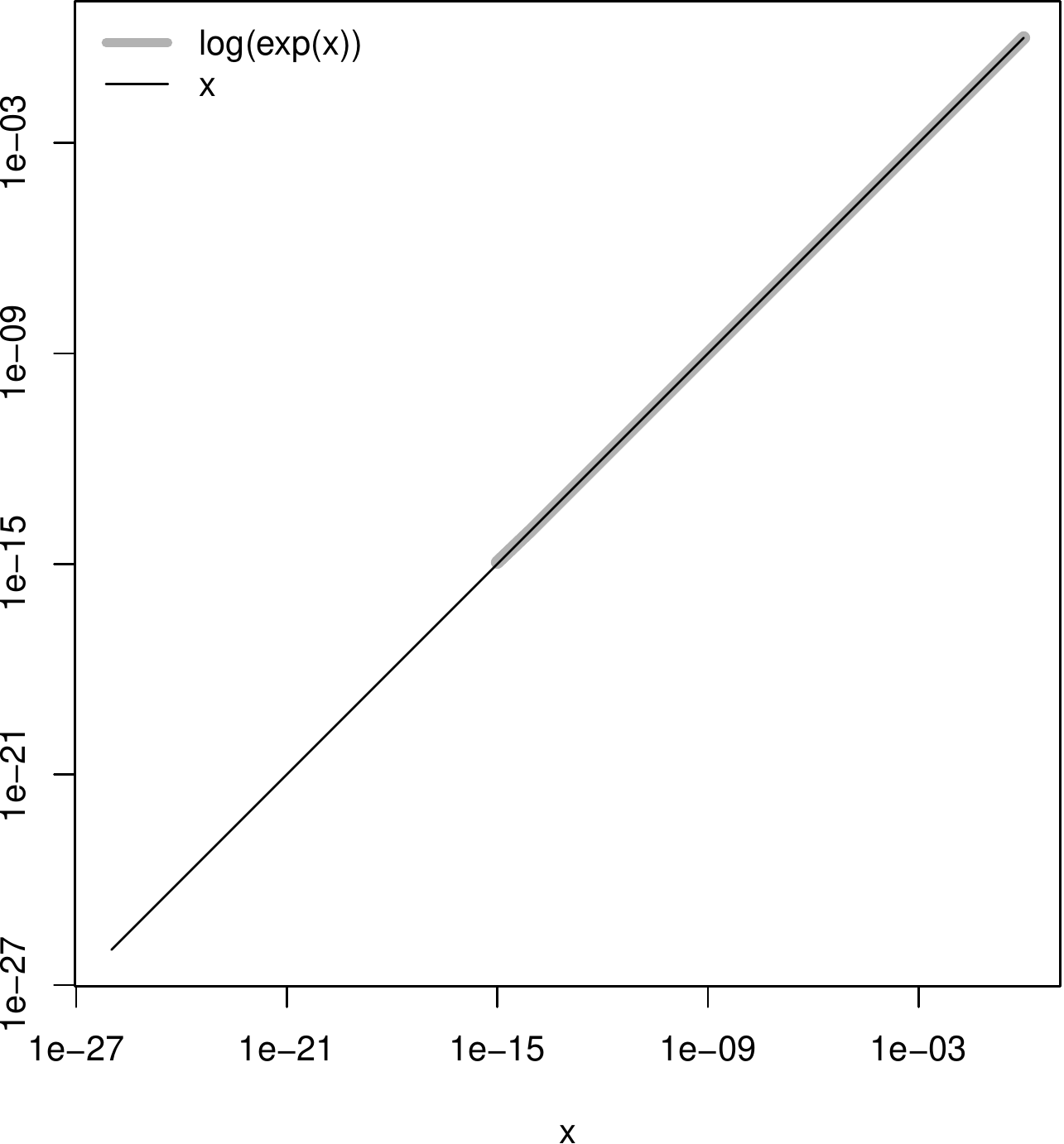} \end{Schunk}
  \end{minipage}
  \caption{Numerical issue when computing $\exp(x)-1$ near $0$ and how \code{expm1()} worked in \R\ version 3.4.4 (left). Numerical issue when computing $\log(\exp(x))$ near $0$ (right).}
\label{fig:expm1:log:exp}
\end{figure}

Also when generating (non-uniform) random variates, numerical issues can appear, for example,
when generating random variates from a gamma distribution. The smaller the shape parameter,
the more probability mass is pushed towards $0$ at which point those random variates smaller
than the smallest positive representable number get truncated to $0$ (and thus produce collisions):
\begin{Schunk}
\begin{Sinput}
> set.seed(271)
> rgamma(10, shape = 1e-3) == 0
\end{Sinput}
\begin{Soutput}
 [1]  TRUE  TRUE FALSE FALSE FALSE FALSE  TRUE FALSE  TRUE  TRUE
\end{Soutput}
\end{Schunk}

\printbibliography[heading=bibintoc]
\end{document}

%
%
%
%
